\documentclass[acmsmall,nonacm,natbib=false]{acmart}

\usepackage[utf8]{inputenc}

\usepackage{amsmath}
\usepackage{amsthm}
\usepackage{enumerate}
\usepackage{xspace}
\usepackage{mathtools}
\usepackage{microtype}
\usepackage{algorithm}
\usepackage{stmaryrd}
\usepackage[style=ACM-Reference-Format,mincrossrefs=9999,giveninits]{biblatex}
\renewcommand\bibfont\footnotesize
\AtEveryBibitem{
  \clearfield{note}
  \clearlist{address}
  \clearlist{location}
  \clearname{editor}
}

\DeclareMathOperator{\val}{val}
\DeclareMathOperator{\Li}{Li}
\DeclareMathOperator{\Id}{Id}
\DeclareMathOperator{\slice}{slice}

\newcommand{\ZZ}{\mathbb Z}
\newcommand{\Zp}{\ZZ_p}
\newcommand{\QQ}{\mathbb Q}
\newcommand{\Qp}{\QQ_p}
\newcommand{\Qq}{\QQ_q}
\newcommand{\Fp}{\mathbb{F}_p}
\newcommand{\RR}{\mathbb R}

\newcommand{\OK}{\mathcal{O}_K}
\newcommand{\OKe}{\OK^{\textrm{\rm ex}}}
\newcommand{\Ke}{K^{\textrm{\rm ex}}}

\newcommand{\RecMat}{\mathcal M}

\newcommand{\PP}{\mathbb P}
\newcommand{\Zpt}{\Zp[\hspace{-0.5ex}[t]\hspace{-0.5ex}]}

\newcommand\Rexp{R_{\textrm{\rm exp}}}
\newcommand\twoFone{{}_2 F_1}
\newcommand\AH{\textrm{AH}}
\newcommand\calA{\mathcal A}

\renewcommand{\mod}{\;\mathrm{mod}\;}
\newcommand{\mathd}{\mathrm d}
\newcommand{\ddt}{\mathrm d/\mathrm dt}
\newcommand{\softO}{O\tilde{~}}
\newcommand{\bigO}{O}

\newcommand{\flint}{\textsc{flint}\xspace}

\definecolor{answer}{rgb}{0,0.5,0.2}

\allowdisplaybreaks[4]
\begin{document}

\theoremstyle{definition}
\newtheorem{remark}[theorem]{Remark}
\newtheorem{hypothesis}[theorem]{Hypothesis}

\title{Fast evaluation of some $p$-adic transcendental functions}
\subtitle{Working draft, \today}

\author{Xavier Caruso}
  \affiliation{
  \institution{Université de Bordeaux, CNRS, INRIA}
  \city{Bordeaux}
  \country{France}
  }
  \email{xavier.caruso@normalesup.org}

\author{Marc Mezzarobba}
  \affiliation{
  \institution{LIX, CNRS, École polytechnique, Institut polytechnique de Paris}
  \postcode{91200}
  \city{Palaiseau}
  \country{France}}
  \email{marc@mezzarobba.net}

\author{Nobuki Takayama}
  \affiliation{
  \institution{Kobe University}
  \city{Kobe}  
  \country{Japan}
  }
  \email{takayama@math.kobe-u.ac.jp}

\author{Tristan Vaccon}
  \affiliation{Universit\'e de Limoges;
  \institution{CNRS, XLIM UMR 7252}
  \city{Limoges}
  \country{France}
  }
  \email{tristan.vaccon@unilim.fr}

\thanks{%
MM's work is supported in part by \grantsponsor{FRANR}{ANR}{https://anr.fr/}
grants \grantnum{FRANR}{ANR-19-CE40-0018} DeRerumNatura
and \grantnum{FRANR}{ANR-20-CE48-0014-02} NuSCAP.
XC's work is supported in part by \grantsponsor{FRANR}{ANR}{https://anr.fr/}
grant \grantnum{FRANR}{ANR-18-CE40-0026-01} CLap--CLap.
TV's work is supported in part by CNRS-INSMI-PEPS-JCJC-2019 grant Patience.}

\begin{abstract} 
We design algorithms for computing values of many $p$-adic 
elementary and special 
functions, including logarithms, exponentials, polylogarithms, 
and hypergeometric functions.
All our algorithms feature a quasi-linear complexity with respect
to the target precision and most of them are based on an adaptation
to the $p$-adic setting of the binary splitting and bit-burst strategies.
\end{abstract}


\vspace{-1.5mm}
\terms{Algorithms, Theory}


\maketitle

\section{Introduction}

Special functions of a complex variable
play a pivotal role in numerous questions arising from analysis, 
geometry, combinatorics and number theory. Examples include the link
between the Riemann $\zeta$-function and the 
distribution of prime numbers, or the Birch and Swinnerton-Dyer
conjecture which relates special values of $L$-functions to
arithmetical invariants of elliptic curves.
Being able to evaluate these functions at high precision is invaluable
for computing invariants or testing conjectures, and work on fast
algorithms for this task over the last decades often makes it possible
nowadays to reach accuracies in the millions of
digits~\cite[\emph{e.g.},][]{Johansson2017}.

At the same time, mathematicians have realized that many complex special 
functions have interesting $p$-adic analogues. 
A famous example is that of 
$p$-adic $L$-functions, which encode subtle invariants 
of towers of number fields (\emph{via} Iwasawa's theory) and, more 
generally, of algebraic varieties.
The algorithmic counterpart of these questions also has attracted
some interest. Efficient algorithms have been designed for computing
the Morita $p$-adic $\Gamma$-function~\cite[\S 6.2]{Villegas2007} and, more recently,
$p$-adic hypergeometric functions~\cite{Asakura2020,Kedlaya2019}
and some $p$-adic $L$-functions~\cite{Belabas2021}.
On a different but closely related note, since the pioneering works of 
Kedlaya~\cite{Kedlaya2001}, much effort has been devoted to computing the 
matrix of the Frobenius acting on the cohomology of $p$-adic algebraic 
varieties~\cite[\emph{e.g.},][]{Lauder2004,Tuitman2019}.

The present paper continues this dynamic and provides new efficient algorithms 
for evaluating many $p$-adic elementary and special functions, including 
polylogarithms, hypergeometric functions and, 
more generally, solutions of ``small'' $p$-adic differential 
equations. In particular, our methods apply to the large class of matrices 
of the Frobenius acting on the cohomology of a fibration, since they
satisfy differential equations of Picard-Fuchs type.

An important feature of our algorithms is that they 
all run in quasi-linear time in the precision.
This contrasts with most previous work
where the complexity was at least quadratic.
The main ingredient for reaching a quasi-optimal complexity is an 
adaptation to the $p$-adic 
setting of the so-called \emph{bit-burst} method introduced by 
Chudnovsky and Chudnovsky~\cite{ChudnovskyChudnovsky1988, 
ChudnovskyChudnovsky1990}, building on the \emph{binary splitting} 
technique~\cite[\emph{e.g.},][]{KoggeStone1973} (see also 
\cite[§178]{BeelerGosperSchroeppel1972}) and other ideas dating back to 
Brent's work on elementary functions~\cite{Brent1976}.
Our algorithms also incorporate later improvements from
\cite{Hoeven2001,Mezzarobba2010,Mezzarobba2011}.
We refer to Bernstein's survey~\cite[esp.~§12]{Bernstein2008} for
more history of the development of these techniques and
further references.

Our starting point is the existence of recurrence relations on the
partial sums of series expansions of the functions we are evaluating.
Roughly speaking, the binary splitting method consists in expressing the
$N$th partial sum as a product of $N$~matrices using this recurrence,
and forming a balanced product tree to evaluate it
(see~§\ref{ssec:ordinary}).
This approach reaches the desired quasi-linear complexity when
the evaluation point $x$ is a small integer.
For more general~$x$, we proceed in several steps:
we find a sequence $x_1, \ldots, x_n = x$ of intermediate 
evaluation points whose bit sizes increase at a
controlled rate while they get closer and closer in the sense of $p$-adic 
distance. Doing this, we can use binary splitting 
to jump from $x_i$ to $x_{i+1}$, and eventually reach $x$
in quasi-linear time.

The remainder of the article is organized as follows.
After preliminaries on the representation of $p$-adic numbers in \S 
\ref{sec:prelim}, we introduce a $p$-adic analogue of bit-burst
method in the simple case of $\log(t)$
in~\S \ref{sec:elementary}. The general case of solutions of 
linear differential equations is addressed in \S \ref{sec:ODE}.
Finally, in \S \ref{sec:examples}, we discuss several applications, including
a fast algorithm for evaluating certain logarithmic derivatives
related to the Dwork hypergeometric functions.

\section{Representation of $p$-adic numbers}
\label{sec:prelim}

Throughout this article, we fix a prime number $p$ and a finite 
extension $K$ of the field of $p$-adic numbers $\Qp$. 
We recall that the $p$-adic valuation on $\Qp$ extends uniquely 
to $K$. We denote it by $\val$ and assume that it is normalized 
by $\val(p) = 1$. We will use the notation $|\cdot|_p$ for the
$p$-adic norm on $K$, defined by
$|x|_p = p^{-\val(x)}$; in particular, we have $|p|_p = p^{-1}$.

\label{ssec:basefield}

It will be convenient to present $K$ as
an unramified extension of~$\Qp$ followed by a totally
ramified extension given by an Eisenstein polynomial.
Let us briefly recall how this works.
We fix a uniformizer $\pi$ of $K$, that is, an element of minimal
positive valuation, and introduce, with $k$ being the residue field of $K$:
\begin{itemize}
\item
the ramification index $e$ of $K/\Qp$ defined by $e = \frac 1{\val(\pi)}$,
\item
the residual degree $f$ of $K/\Qp$ defined by $f = [k:\Fp]$.
\end{itemize}
We choose a monic polynomial $U(X) \in \ZZ[X]$ of degree $f$ 
whose reduction modulo $p$ is irreducible.
One easily checks that $U(X)$ remains irreducible in $\Qp[X]$
and we can then form the field $\Qq = \Qp[X]/U(X)$.
It follows from Hensel's lemma that $\Qq$ embeds (non-canonically) 
into $K$. Let now $V(Y) \in \Qq[Y]$ be the minimal polynomial of~$\pi$
over $\Qq$. One can show that $V(Y)$ is an Eisenstein polynomial,
and in particular that it lies in $\Zp[X,Y]/U(X)$.
Besides, Krasner's lemma~\cite[§3.1.5]{Robert2013} indicates that
we can assume (up to changing $\pi$) that $V(Y) \in
\ZZ[X,Y]/U(X)$.
Thus, viewing~$V$ as a bivariate polynomial over $\ZZ$,
we have the presentations:
\begin{equation}
K \simeq \Qp[X,Y]/(U, V), \label{eq:presK}
\quad
\OK \simeq \Zp[X,Y]/(U, V)
\end{equation}
where $\OK$ denotes the ring of integers of $K$ (which consists
of the elements of nonnegative valuation).
If $a 
\in K$ is represented by the polynomial
$\sum_{i=0}^{f-1} \sum_{j=0}^{e-1} a_{ij} X^i Y^j$ (with
$a_{ij} \in \Qp$), one has:
\begin{align}
\label{eq:val}
\val(a) & = 
\displaystyle \min_{i,j} \big(\val(a_{ij}) + {\textstyle \frac j e}\big), &
|a|_p & = 
\displaystyle \max_{i,j} \big(|a_{ij}|_p \cdot p^{-j/e}\big).
\end{align}

\label{ssec:exact subring}

The presentation of $K$ we have picked allows us
to define a canonical \emph{exact} subrings $\Ke$ and $\OKe$ by
(compare with Eq.~\eqref{eq:presK})
\[
\Ke \simeq \QQ[X,Y]/(U, V),
\quad
\OKe \simeq \ZZ[X,Y]/(U, V).\]
The ring $\OKe$ is dense in the ring of integers $\OK$ of $K$,
which concretely means that given an element $a
\in K$ with nonnegative valuation and an integer $n$, one can 
always find $b \in \OKe$ such that $a \equiv b \pmod{p^n}$.
Similarly, the ring $\Ke$ is dense in $K$.
Additionally, one can define a \emph{height} function $h : \OKe \to \RR^+$
which measures the bit size of the elements by
\begin{equation}
  \label{eq:hK}
  h\Big( 
  {\textstyle \sum_{i=0}^{f-1} \sum_{j=0}^{e-1} a_{ij} X^i Y^j}
  \Big) = \max_{i,j}\, \log\big(1 + |a_{ij}|\big)
\end{equation}
where the coefficients $a_{ij}$ are integers and the notation
$|a_{ij}|$ refers to the usual absolute value.
If $x$ and $y$ are elements of $\OKe$ of height bounded by 
$H$, one can compute the sum $x + y$ for a cost of $O(H)$ bit 
operations.
Similarly, one can compute the product $xy$ and reduce it for a total
cost of $\softO(H)$ bit operations (where the hidden constant depends on
$U,V$ and hence on $e,f$).

\begin{proposition}
\label{prop:height}
Given a choice of defining polynomials $U$~and~$V$,
there exists a constant $C \geq 0$ such that the 
function $h_K = h + C$ satisfies
\begin{align*}
h_K(x_1 + \cdots + x_s) & \leq \max(h_K(x_1), \ldots, h_K(x_s)) + 
\log(s), \\
h_K(x_1 x_2 \cdots x_s) & \leq h_K(x_1) + \cdots + h_K(x_s),
\end{align*}
for all $x_1, \ldots, x_s \in \OKe$.
\end{proposition}

\begin{proof}
The first inequality follows (for any choice of~$C$)
from the observation that
$|a_1 + \cdots + a_s| \leq
s \cdot \max(|a_1|, \ldots, |a_s|)$
when $a_1, \ldots, a_s$ are integers.
In order to prove the second inequality, it is enough
to check that $h_K(xy) \leq h_K(x) + h_K(y)$, \emph{i.e.} $h(xy)
\leq h(x) + h(y) + C$ for some constant~$C \geq 0$ and all $x, y \in 
\OKe$.
Using bilinearity of the product and the first part of the
proposition, one finds that one can take
$C = \log(ef) + \max_{\substack{0 \leq i \leq 2f-2 \\ 0 \leq j \leq 2e-2}}
h\big(X^i Y^j \big)$.
\end{proof}

We fix a function $h_K : \OKe \to \RR^+$ 
satisfying the requirements of Proposition~\ref{prop:height}.
It is advisable to minimize~$C$ because its value has a direct
impact on the complexity of our algorithms.
The proof of Proposition~\ref{prop:height}
shows that the value of~$C$ is related to the degrees and heights
of the polynomials $U$ and $V$.
Since $U$ is defined as a lift of a polynomial over $\Fp$, one can
always assume $h(U) \leq \log p$.
Bounding the height of $V$ is more complicated but, by Krasner's lemma,
it reduces to bounding the ramification of $K$, a problem that can be
attacked using Newton polygons techniques.

\section{Elementary functions}
\label{sec:elementary}

\subsection{Logarithm}
\label{ssec:log}

Let us start with the most basic transcendental functions,
namely the $p$-adic logarithm and
exponential~\cite[\emph{e.g.},][§4.5]{Robert2013}.

On the open unit disk centered at $1$, the $p$-adic logarithm is defined
by the usual convergent series
$\log(1-t) = - \sum_{i=1}^\infty t^i/i$.
Given $x \in K$, $|x|_p < 1$, our aim is to compute $\log (1-x)$ 
efficiently at high precision.
The algorithm we describe is a straightforward adaptation of one of the
classical algorithms for the same task over the reals.
That the idea generalizes to~$\Qp$ is folklore
(it is implemented in \flint\ and a special case is mentioned
in~\cite{Bernstein2008})
but, to our knowledge, no full analysis in the $p$-adic setting
appears in the literature.
Since the ideas behind this algorithm will return
in the next sections, we discuss it in some detail.

First of all, it is useful to know how accurate the input has to be
to determine $\log (1-x)$ to an accuracy $O(\pi^\sigma)$.

\begin{lemma}
\label{lem:preclog}
Let $u, v \in K$ with $|u|_p < 1$ and $|v|_p < 1$. If one has
$u = v + O(\pi^m)$
for an integer $m \geq e/(p{-}1)$, then
$\log(1-u)  = \log(1-v) + O(\pi^m).$
\end{lemma}

\begin{proof}
Writing $u = v + \pi^m w$ with $|w|_p \leq 1$ and expanding
$u^p = (v + \pi^m w)^p$, we deduce that $u^p \equiv v^p \pmod
{p \pi^m}$. Repeating the same argument, we find that
$u^{p^s} \equiv v^{p^s} \pmod {p^s \pi^m}$ for all $s \geq 0$.
Therefore $\frac{u^i}i \equiv \frac{v^i}i \pmod {\pi^m}$ for
all positive integer $i$ and the result follows from
the definition of~$\log(1-t)$.
\end{proof}

For the actual computation of $\log(1-t)$,
we use the \emph{``digit-burst''} strategy
materialized by the following lemma, in which $C$ 
denotes the constant of Proposition~\ref{prop:height}.

\begin{lemma}
\label{lem:bitburstlog}
Given $x \in K$, $|x|_p < 1$, there exists a decomposition:
\[1 - x \equiv (1 - x_1) \cdot (1 - x_2) \cdots (1 - x_\ell)
\pmod{\pi^\sigma}\]
with $\ell = \lceil \log_2 \frac \sigma e \rceil$ and
for all $s \in \{1, \ldots, \ell\}$, $x_s \in \OKe$ such that
\[ \textstyle
  \val(x_s) > 0,\,
  \val(x_s) \geq 2^{s-1} - 1
  \,\,\text{and}\,\,
  h_K(x_s) \leq (2^s{-}1) \log p + C.
\]
\end{lemma}

\begin{proof}
We recall that $x$ is represented by a polynomial of the form
$\sum_{i} \sum_{j} a_{ij} X^i Y^j$
with $a_{ij} \in \Zp$ and $a_{00} \in p \Zp$ (see \S \ref{ssec:basefield}).
We define $x_1$ by
$x_1 = \sum_{i} \sum_{j} a'_{ij} X^i Y^j$
where $a'_{ij}$ is the unique integer in $[0, p{-}1]$ which is
congruent to $a_{ij}$ modulo $p$.
Using \eqref{eq:val}~and~\eqref{eq:hK}, one has
$\val(x_1) > 0$ and $h_K(x_1) \leq \log p + C$.

The quotient $(1{-}x)/(1{-}x_1)$ is congruent to $1$ modulo $p$,
hence it is represented by a polynomial of the form
$\sum_{i} \sum_{j} b_{ij} X^i Y^j$
with $b_{00} \equiv 1 \pmod p$ and $b_{ij} \equiv 0 \pmod p$
for $(i,j) \neq (0,0)$. We set
$1 - x_2 = \sum_{i} \sum_{j} b'_{ij} X^i Y^j$
where $b'_{ij} \in [0, p^3{-}1]$ is congruent to $b_{ij}$ modulo 
$p^3$. One has
$1-x = (1-x_1){\cdot}(1-x_2){\cdot}(1+O(p^3))$
with
$\val(x_2) \geq 1$ and  $h_K(x_2) \leq 3 \log p + C$.

Repeating this process $\ell$ times, we obtain the lemma.
\end{proof}

\begin{lemma}
\label{lem:binarylog}
Let $u \in \OKe$ with $|u|_p < 1$. One can compute $\log(1-u)$ 
modulo $\pi^\sigma$
For a cost of $\softO\big(\sigma \cdot\frac{h_K(u)}{\val(u)} \big)$ 
bit operations.
\end{lemma}

\begin{proof}
We have $\val\big(\frac {u^i} i\big) = 
i \val(u) - \val(i) \geq i \val(u) - \log_p i$.
Thus
$\log(1-u) \equiv \sum_{i=1}^N \frac {u^i} i \pmod {\pi^\sigma}$
if $N \val(u) - \log_p N \geq \frac{\sigma}{e}$. This occurs as soon
as $N = \softO\big(\frac{\sigma}{\val(u)} \big)$. 
Since the numerator (in~$\OKe$) and denominator (in~$\ZZ$) of any sum of
the form $\sum_{i=1}^n \frac{\smash{u^i}}{i_0+i}$ have height at most
$n (h_K(u) + \log(N+1)) + \log n)$,
one can compute the \emph{exact} value of the 
finite sum $\sum_{i=1}^N \frac {u^i} i \in \Ke$ in $\softO(N h_K(u))$ bit 
operations using a divide-and-conquer strategy. The lemma follows.
\end{proof}

Putting everything together, we get the following theorem.

\begin{theorem} \label{thm:log}
There exists an algorithm that takes as input an element $x
\in K$, $|x|_p < 1$ and outputs $\log(1-x)$ 
at precision $O(\pi^\sigma)$ 
for a cost of $\softO(\sigma)$ bit operations.
\end{theorem}

\begin{proof}
Without loss of generality, one may assume $\sigma \geq \frac e{p-1}$.
Combining Lemmas~\ref{lem:preclog} and~\ref{lem:bitburstlog}, we find a
congruence of the form:
\[\log(1 - x) \equiv \log(1 - x_1) + \cdots + \log (1 - x_\ell) 
\pmod {\pi^\sigma}\]
with $\val(x_s) \geq 2^{s-1}-1+\frac{1}{e}$ and $h_K(x_s) \leq (2^s-1) \log p + C$ for 
all $s \in \{1,\ldots,\ell\}$. By Lemma~\ref{lem:binarylog}, each summand
$\log(1 - x_s)$ can be evaluated for a cost of
$\softO\big(\sigma \cdot\frac{h_K(x_s)}{\val(x_s)} \big)
\subset \softO(\sigma)$ bit operations.
Since $\ell$ itself stays within
$O(\log \sigma)$, the theorem is proved.
\end{proof}

It is possible to track the dependency in the field~$K$ of the
complexity through the proof.
Doing this, we obtain a total cost of
$\softO(\sigma\cdot(C + \log p))$ where $C$ is the constant of
Proposition~\ref{prop:height} and the constants hidden in the
$\softO$ are now absolute.

\subsection{Exponentiation}
\label{ssec:exp}

\subsubsection*{The exponential function.}

The $p$-adic exponential function is the function defined by
$\exp(t) = \sum_{i=0}^\infty t^i/i!$.
Using Legendre's formula, one shows that its radius of convergence is
$\Rexp = p^{-1/(p-1)} < 1$ and that it assumes values in the 
open unit disk centered at $1$.

Let $x \in K$ with $|x|_p < \Rexp$. We aim at computing $\exp(x)$ at 
precision $O(\pi^\sigma)$ in time $\softO(\sigma)$.
It is possible to use a similar ``digit-burst'' technique 
as for the logarithm. Instead, we present a different approach that we 
will reuse later on:
we solve the equation $\log y = x$ (of unknown $y$) 
using a Newton scheme.

For this, we consider the function~$f$ defined for $|y-1|_p < 1$ by
$f(y) = x - \log y$.
The Newton iteration formula associated to $f$ is
\begin{equation}
\label{eq:newtonexp}
y_{s+1} = y_s - \frac{f(y_s)}{f'(y_s)} = 
y_s \cdot (1 + x - \log y_s).
\end{equation}
Any sequence $(y_s)_{s \geq 0}$ satisfying~\eqref{eq:newtonexp}
will rapidly converge to $\exp(x)$ provided that $y_0$ 
is close enough to $\exp(x)$.
Noticing that $|f'(y)|_p = |f''(y)|_p = 1$ as soon as $|y{-}1|_p < 1$,
we deduce from \cite[Cor.~3.2.14]{Caruso2017}
(applied with $C = \Rexp^{-1}$)
that a sufficient condition for convergence is 
$|y_0 - \exp(x)|_p < \smash{\Rexp}$. Now, observe that:
\[\textstyle 
\val\big(\frac{x^i}{i!}\big) \geq i \cdot \big(\!\val(x) - \frac 1{p-1}\big)
\geq \frac i{e(p-1)}\]
the second inequality following from the fact
that $\val(x) - \frac 1{p-1}$ is a positive element of 
$\frac{1}{e(p-1)} \ZZ$. Hence, one can start the Newton
iteration with
$y_0 = \sum_{i=0}^{e-1} \frac{\smash{x^i}}{i!} \text{ mod } \pi^m$
where $m$ is any integer strictly greater that $\frac e {p-1}$.
The cost of the computation of $y_0$ is independent of the
target precision $\sigma$. 
Finally, \cite[Cor.~3.2.14]{Caruso2017} tells us that $y_s \equiv
\exp(x) \pmod {\pi^\sigma}$ provided that 
$2^s \big(\! \val\big(y_0 - \exp(x)\big) - \smash{\frac 1 {p-1}}\big) 
\geq \frac \sigma e$,
which holds for $2^s \geq (p{-}1) \sigma$.
One can take $s = O(\log \sigma)$, proving that the total cost 
of the algorithm is $\softO(\sigma)$.

\subsubsection*{Powering}
\label{ssec:pow}

Given $\delta \in \Zp$ and $x$ in the open unit disk of $K$, one 
can give a meaning to the expression $(1+x)^\delta$ by setting:
\begin{equation}
  \label{eq:defpow}
  (1+x)^\delta = \sum_{i=0}^\infty
  \frac {\delta(\delta-1) \cdots (\delta-i+1)}{i!} x^i.
\end{equation}
A first idea to compute this value in quasi-optimal complexity is to
write
$(1+x)^\delta = \exp\big( \delta \cdot \log(1+x)\big)$.
However, it does not quite work because the latter equality only
makes sense when $\delta \cdot \log(1+x)$ falls inside the disk
of convergence of the exponential.
Instead, we observe that $y = (1+x)^\delta$ always satisfies the
equation
\[\log y = \delta \cdot \log(1+x)\]
and solve it using a Newton scheme as we did in
\S \ref{ssec:exp}: we start by computing a first rough approximation
$y_0$ of $(1+x)^\delta$ and then iterate the Newton
operator~\eqref{eq:newtonexp}. As before, the precision we need on
$y_0$ is $O(p^{1/(p{-}1)})$, so that we can take
the series~\eqref{eq:defpow} truncated after
$m = \lceil e/(p{-}1) \rceil$ terms for~$y_0$.
The total complexity of the computation of $(1+x)^\delta$ 
is at most $\softO(\sigma)$.

The same strategy applies to the \emph{Artin-Hasse exponential}
$\AH(t) = \exp(t + p^{-1} t^p + \cdots + p^{-n} t^{p^n} + \cdots)$,
a useful renormalization of the $p$-adic exponential with a larger
radius of convergence~\cite[§7.2]{Robert2013}.

\section{Solution of differential equations}
\label{sec:ODE}

Our goal is now to generalize the previous results to the evaluation of
a large class of solutions of differential equations.
We consider a linear differential equation of the form
\begin{equation}
\label{eq:ODE}
a_r(t) y^{(r)}(t) + \cdots +
a_1(t) y'(t) + a_0(t) y(t) = 0
\end{equation}
where the~$a_i$ are polynomials of degree at most~$d$, with coefficients
in~$\OKe$ of height at most~$\ell$.
Substituting $y (t) = \sum_{n \geq 0} y_n t^n$ into~\eqref{eq:ODE} shows that
the coefficient sequence~$(y_n)$ of a formal power series solution~$y$ must
satisfy
\begin{equation}
\label{eq:rec}
  b_0 (n) y_n + b_1 (n) y_{n - 1} + \cdots + b_s (n) y_{n - s} = 0,
\end{equation}
\begin{equation}
\label{eq:reccoeff}
  b_j (n) = \sum_{i = 0}^r a_{i, i - r + j}  (n - j)  (n - j - 1) \cdots (n -
  j - i + 1),
\end{equation}
where $s = r + d$ and
$a_{i, j}$ is the coefficient of $t^j$ in~$a_i (t)$ ($0$ if $j < 0$). In
particular, one has
\begin{equation}
\label{eq:ind ordinary}
  b_0 (n) = a_r (0) n (n - 1) \cdots (n - r + 1) . 
\end{equation}
The recurrence~{\eqref{eq:rec}} holds for all~$n \in \mathbb{Z}$ if the 
sequence~$(y_n)$ is extended by $y_n = 0$ for $n < 0$. 

\subsection{Partial sums at ordinary points}
\label{ssec:ordinary}

The recurrence~\eqref{eq:rec} can be used to evaluate partial sums of the series~$y (t)$
efficiently by binary splitting. We recall and analyze an algorithm for this task,
essentially the ``optimized'' version from~{\cite{Mezzarobba2010}} of a method
first detailed in~{\cite[§5--6]{ChudnovskyChudnovsky1988}}.

We assume in this subsection that $a_r(0) \neq 0$.
Then, the space of formal power series solutions of~\eqref{eq:ODE} has
dimension~$r$ and admits a basis
$(f_0, \ldots, f_{r - 1})$
such that $f_j (t) = t^j + \bigO (t^r)$.
We denote by
$\Phi_0(t) = \bigl(\frac{1}{i!} f_{\smash j}^{\smash{(i)}}(t)\bigr)$
the associated fundamental matrix.
There exists $\rho > 0$ such that the~$f_j$
converge on the open disk of radius $\rho$ centered at $0$.
For future use, we also define
$\Phi_\xi(t) = \bigl(\frac{1}{i!} g_j^{\smash{(i)}}(\xi+t)\bigr)$
at an arbitrary~$\xi$ with $a_r(\xi) \neq 0$,
where $g_j$ now is the solution such that
$g_j (\xi+t) = t^j + \bigO (t^r)$.

We are given an integer~$N$ and an element $x \in \Ke$, written in the
form $x = u / v$ with $u \in \OKe$ and $v \in \mathbb{Z}$, and our
task is to compute the $N$-th partial sum of~$\Phi_0(x)$.
The general idea of the algorithm is to encode the simultaneous
computation of the entries of~$\Phi_0(x)$ in a product of matrices that
is computed in rational arithmetic by forming a balanced product tree.
For space reasons, we limit ourselves here to a technical description
of the procedure and refer to~\cite{Mezzarobba2011} for more context.

Given a ring~$R$, an indeterminate~$Z$, and $k \in \ZZ_{\geq 0}$, define
the jet space
$J^k_Z (R) = R\llbracket Z\rrbracket / Z^k$.
Observe that computing each column of $\Phi_0(x)$ reduces to evaluating
one of the~$f_j$ at $x+\Delta \in J^r_{\Delta}(\Ke)$.
(Jet spaces in a second indeterminate~$\Lambda$ will be used in \S \ref{ssec:regsing}.)
Let~$\mathcal{M}$ be the set of tuples
$T = (\mathsf{C}_T, \mathsf{d}_T, \mathsf{u}_T, \mathsf{v}_T,
\mathsf{R}_T)$
with $\mathsf{C}_T \in ( \OKe )^{s \times s}$,
$\mathsf{d}_T \in \OKe$, $\mathsf{u}_T \in J^r_{\Delta} ( \OKe )$,
$\mathsf{v}_T \in \mathbb{Z}$, and $\mathsf{R}_T \in J^r_{\Delta} ( \OKe
)^s$. We equip~$\mathcal{M}$ with the product defined by
\begin{equation}
  T T' = (\mathsf{C}_T  \mathsf{C}_{T'}, \mathsf{d}_T  \mathsf{d}_{T'},
  \mathsf{u}_T  \mathsf{u}_{T'}, \mathsf{v}_T  \mathsf{v}_{T'}, \mathsf{R}_T 
  \mathsf{C}_{T'}  \mathsf{u}_{T'} + \mathsf{d}_T  \mathsf{v}_T  \mathsf{R}_{T'} )
  \label{eq:matmul}
\end{equation}
where $\mathsf{R}_T, \mathsf{R}_{T'}$ are viewed as row vectors.
The product is associative.
In fact, multiplying
elements of~$\mathcal{M}$ amounts to multiplying ${(s + 1)} \times (s + 1)$
matrices of the form
\[ \mathsf{M}_T = \begin{pmatrix}
     \mathsf{C}_T  \mathsf{u}_T & 0\\
     \mathsf{R}_T & \mathsf{d}_T  \mathsf{v}_T
   \end{pmatrix}, \]
but the representation~\eqref{eq:matmul} makes fast computations with
these special matrices easier to state and analyze.

For $n \in \mathbb{N}$, let $B (n)$ be the element of~$\mathcal{M}$ defined by
\[ \mathsf{C}_{B (n)} = \left(\begin{matrix}
     0 & b_0 (n) &  & \\
     &  & \ddots & \\
     &  &  & \phantom{+} b_0 (n)\\
     - b_s (n) & \cdots & \cdots & - b_1 (n)
   \end{matrix}\right),
   \quad
   \begin{aligned}
   \mathsf{R}_{B (n)} &= (0, \ldots, 0, vb_0 (n)),\\
   \mathsf{d}_{B (n)} &= b_0 (n), \\
   \mathsf{u}_{B (n)} &= u + v \Delta, \\
   \mathsf{v}_{B (n)} &= v.
   \end{aligned}
\]
Also define
$P (n_0, n_1) = B (n_1 - 1) \cdots B (n_0 + 1) B (n_0)$.

\begin{algorithm}[t]
\caption{$
\begin{aligned}[t]
  \mathrm{PartialSum} \bigl(
  &(a_i) \in \OKe [t]^r, u \in \OKe, v \in \mathbb{Z},\\[-5pt]
  &N \in \mathbb{Z}_{\geq r}, \sigma \in \mathbb{Z}_{\geq 0} \bigr)
\end{aligned}$}
 \label{algo:bs ordinary}
\begin{enumerate}[(1)]
  \item \label{step:compute rec}
  Compute the polynomials $b_j$ defined by~{\eqref{eq:rec}},
  {\eqref{eq:reccoeff}}.
  \item \label{step:bs}
  Compute $\Pi = P (r, N)$ using recursively the formula $P (n_0, n_1) = P
  (m, n_1) P (n_0, m)$ with $m \approx (n_0 + n_1) / 2$.
  \item Extract the last~$r$ columns of $\mathsf{C}_{\Pi}$
  as a matrix $U \in (\OKe)^{s \times r}$.
  \item \label{step:extract}
  Compute and return the matrix
  \[ \tilde{\Phi} = [v^{- j}  \mathsf{d}_{\Pi}^{- 1} \mathsf{v}_{\Pi}^{- 1} 
     ((u + v \Delta)^j  (\mathsf{R}_{\Pi})_{s - j})_i]_{0 \leq i, j < r}
     \mod \pi^\sigma \in K^{r \times r} \]
  where $(\mathsf{R}_{\Pi})_{s - j}$ is the entry of~$\mathsf{R}_{\Pi}$ of
  index $s - j$, counting from zero, and $\xi_i$ for $\xi \in J^r_{\Delta}
  (K)$ is the coefficient of $\Delta^i$ in~$\xi$.
\end{enumerate}
\end{algorithm}

\begin{proposition} \label{prop:partial sums}
  Write
  $f_j (t) = \sum_{n \geq 0} f_{j, n} t^n$
  and let
  $f_j (t)_{< N} = \sum_{n = 0}^{N - 1} f_{j, n} t^n$.
  Algorithm~\ref{algo:bs ordinary} returns
  \[ \left[ \frac{1}{i!}  \left(
     \frac{\mathd^i}{\mathd t^i} (f_j (t)_{< N}) \right)_{t = x} \!\!\! \mod \pi^\sigma \right]_{0 \leq i, j < r}.
  \]
\end{proposition}

\begin{proof}
  The recurrence relation~{\eqref{eq:rec}} translates into
  \[ b_0 (n)  (y_{n - s + 1}, \ldots, y_n)^{\mathrm T} = \mathsf{C}_{B (n)}  (y_{n - s}, \ldots, y_{n - 1})^{\mathrm T}. \]
  Letting $Y_n = (y_{n - s}  \xi^{n-1}, \ldots, y_{n-1}\xi^{n-1}, y (\xi)_{< n})^{\mathrm T}$
  where $\xi = x + \Delta \in J^r_{\Delta} (K)$ and $y (t)_{< N}$ denotes the
  partial sum $\sum_{n = 0}^{N - 1} y_n t^n$, one has
  \[ Y_n = \left(\begin{matrix}
       0 & \xi &  &  & 0\\
       &  & \ddots &  & \vdots\\
       &  &  & \xi & \vdots\\
       c_s(n) \xi & \cdots & \cdots & c_1(n) \xi & 0\\
       0 & \cdots & 0 & 1 & 1
     \end{matrix}\right) Y_{n - 1}, \quad
     c_i(n) = -\frac{b_i(n)}{b_0(n)},
  \]
  that is,
  $Y_n = (\mathsf{v}_{B (n)}  \mathsf{d}_{B (n)})^{- 1}  \mathsf{M}_{B (n)} Y_{n - 1}$,
  whenever~$b_0 (n) \neq 0$. Taking into account~{\eqref{eq:ind ordinary}}, it
  follows that
  \[ (y_{N - s + 1}, \ldots, y_N)^{\mathrm T} = \mathsf{d}_{\Pi}^{- 1} 
     \mathsf{C}_{\Pi}  (0, \ldots, 0, y_0, \ldots, y_{r-1})^{\mathrm T}, \]
  and
  $Y_N = (\mathsf{v}_{\Pi}  \mathsf{d}_{\Pi})^{- 1}  \mathsf{M}_{\Pi} Y_{r-1}$.
  When $y$ is set to the element~$f_j$ of the distinguished basis defined
  above, the corresponding initial vector is
  $Y_r = ([\text{$d{+}j$ zeros}], u^{r-1}/v^{r-1}, [\text{$r{-}j{-}1$ zeros}], u^j / v^j))^{\mathrm T}$.
  This leads to the formulas used at step~\ref{step:extract}.
\end{proof}

We turn to the complexity analysis.
When $A$ is an $\OKe$-algebra equipped with a distinguished $\OKe$-basis
$(e_i)$, such as $A = \OKe[n]$, we extend $h_K$ to~$A$ by setting
$h_K (\sum_i \lambda_i e_i) = \max_i h_K (\lambda_i)$.

\begin{lemma} \label{lem:height rec}
  For $j=0, \dots, r$, the coefficient~$b_j$ of the recurrence
  relation~\eqref{eq:rec} satisfies
  $h_K(b_j) \leq \ell + (s + 1) \log s$.
\end{lemma}

\begin{proof}
  The coefficients of the polynomials
  $(n{-}j) \cdots (n{-}j{-}i{+}1)$
  appearing in~{\eqref{eq:reccoeff}} are all bounded by
  $(1{+}j) \cdots (1{+}j{+}i{-}1) \leq s!$ since $i{+}j \leq s$
  in all the terms.
  Coming back to the definition of $h_K$, we get
  $h_K(b_j) \leq \log(s!) + \max_{0 \leq i \leq r}
  h_K(a_{i,i-r+j}) \leq s \log s + s + \ell$.
\end{proof}

\begin{lemma} \label{lem:bs heights}
  Let~$H$ be such that $h_K (u), h_K (v) \leq H$. Let $0 \leq n_0 <
  n_1$ and $\Pi = P (n_0, n_1)$. We have the bounds
  \begin{align*}
    h_K (\mathsf{u}_{\Pi}), h_K (\mathsf{v}_{\Pi}) &\leq (n_1 - n_0)  (H
    + \log (s)), \\
    h_K (\mathsf{C}_{\Pi}), h_K (\mathsf{d}_{\Pi}) &\leq (n_1 - n_0) 
    (\ell + (s + 2)  (\log s + \log n_1) + 1), \\
    h_K (\mathsf{R}_{\Pi}) &\leq (n_1 - n_0)  (\ell + H + (s + 3)  (\log
    s + \log n_1) + 1).
  \end{align*}
\end{lemma}

\begin{proof}
  We start by bounding the height of the elements of~$B (n)$ for $n \leq
  n_1$.
  We have by assumption
  $h_K (\mathsf{u}_{B (n)}), h_K (\mathsf{v}_{B (n)}) \leq H$.
  Lemma~\ref{lem:height rec} implies
  $h_K (b_j (n)) \leq \beta(n) := \ell + (s + 1) \log s + d \log n + 1$
  for all $j$ and $n$.
  We hence have
  $h_K (\mathsf{C}_{B (n)}), h_K (\mathsf{d}_{B (n)}) \leq \beta (n)$
  and
  $h_K (\mathsf{R}_{B (n)}) \leq \beta (n) + H$.
  For $T, T' \in \RecMat$, Proposition~\ref{prop:height} yields
  \begin{align*}
    h_K (\mathsf{C}_{TT'}) &\leq h_K (\mathsf{C}_T) + h_K
    (\mathsf{C}_{T'}) + \log (s), \\
    h_K (\mathsf{d}_{TT'}) &\leq h_K (\mathsf{d}_T) + h_K
    (\mathsf{d}_{T'}), \\
    h_K (\mathsf{v}_{TT'}) &\leq h_K (\mathsf{v}_T) + h_K
    (\mathsf{v}_{T'}), \\
    h_K (\mathsf{u}_{TT'}) &\leq h_K (\mathsf{u}_T) + h_K
    (\mathsf{u}_{T'}) + \log (r) .
  \end{align*}
  The first inequality implies
  \[ h_K (\mathsf{C}_{T_1 \cdots T_m}) \leq h_K (\mathsf{C}_{T_1}) +
     \cdots + h_K (\mathsf{C}_{T_m}) + m \log (s) \]
  whence
  $h_K (\mathsf{C}_{\Pi}) \leq (n_1 - n_0)  (\beta (n_1) + \log (s))$.
  The height bounds on $\mathsf{d}_{\Pi}$, $\mathsf{u}_{\Pi}$, and
  $\mathsf{v}_{\Pi}$ follow in the same way.
  
  In the case of~$\mathsf{R}_{\Pi}$, Proposition~\ref{prop:height} implies
  \[ h_K (\mathsf{R}_{TT'}) \leq \max (h_K (\mathsf{R}_T  \mathsf{C}_{T'}
     \mathsf{u}_{T'}), h_K (\mathsf{d}_T  \mathsf{v}_T  \mathsf{R}_{T'})) .
  \]
  With $T = B (n)$ and $T' = P (n - m, n - 1)$, one has
  \begin{align*}
    h_K (\mathsf{R}_T  \mathsf{C}_{T'}  \mathsf{u}_{T'})
    & \leq h_K (\mathsf{R}_T) + h_K (\mathsf{C}_{T'}) + h_K
    (\mathsf{u}_{T'}) + \log (s) + \log (r) \\
    &\leq \beta (n) + H + (m - 1)  (\beta (n) + \log (s) + 
    H + \log (r)) \\
    &\leq m (\beta (n) + H + 2 \log (s)), \\
    h_K (\mathsf{d}_T  \mathsf{v}_T  \mathsf{R}_{T'})
    &\leq h_K (\mathsf{d}_T) + h_K (\mathsf{v}_T)
    + h_K (\mathsf{R}_{T'}) + \log (r) \\
    &\leq \beta (n) + H + h_K (\mathsf{R}_{T'}) + \log (r).
  \end{align*}
  Thus
  $h_K (\mathsf{R}_{TT'}) \leq \max (h_K (\mathsf{R}_{T'}) + \gamma
     (n), m \gamma (n))$
  where $\gamma (n) = \beta (n) + H + 2 \log (s)$.
  By induction on~$m$, one gets
  $h_K (\mathsf{R}_{\Pi}) \leq (n_1 - n_0) \gamma (n_1)$.
  The claim follows.
\end{proof}

\begin{proposition} \label{prop:bs complexity}
For $u$, $v$ of height $\leq H$ and $\sigma = \softO(N)$,
Algorithm~\ref{algo:bs ordinary} runs in
$\softO\bigl( s^{\omega} N (\ell + H + s) \bigr)$
bit operations, where $\omega$ is the exponent of matrix multiplication.
\end{proposition}

\begin{proof}
  The bulk of the cost comes from step~\ref{step:bs}.
  Step~\ref{step:bs} decomposes into the computation of the
  tuples~$B(n)$ for $n = r, \dots, N$, and the construction of a product
  tree from these tuples.
  The evaluation of $b_j$ at~$n$ can be performed in
  $\softO(h_K(b_j) + r \log n)= \softO(\ell + s + \log n)$
  operations in a divide-and-conquer fashion~\cite{Estrin1960,BostanCluzeauSalvy2005},
  leading to a total cost of
  $\softO(s (\ell + s + \log n) + H)$
  for the construction of each~$B(n)$.

  Consider two subproducts
  $\Pi_0 = P(n_0,n_1)$ and $\Pi_1 = P(n_1,n_2)$
  with $0 \leq n_1{-}n_0,\, n_2{-}n_1 \leq m$.
  Using the bounds from Lemma~\ref{lem:bs heights} and standard
  bounds on the complexity of arithmetic in $\mathbb Z[X]$,
  one sees that computing
  $\mathsf{u}_{\Pi_1\Pi_0}$ and $\mathsf{v}_{\Pi_1\Pi_0}$
  from $\Pi_0$, $\Pi_1$ takes
  $\softO(mr(H + \log s)) = \softO(msH)$
  operations.
  Similarly, the computation of
  $\mathsf{C}_{\Pi_1\Pi_0}$
  and
  $\mathsf{d}_{\Pi_1\Pi_0}$
  requires a total of
  $\softO(m(s^{\omega+1} \ell \log N))$
  operations.
  Finally, one can compute
  $\mathsf{d}_{\Pi_1}\mathsf{v}_{\Pi_1}\mathsf{R}_{\Pi_0}$
  in
  $\softO(msr\cdot(\ell + H + s \log N))$
  operations, and
  $\mathsf{R}_{\Pi_1}\mathsf{C}_{\Pi_0}\mathsf{u}_{\Pi_0}$
  in
  $\softO(m (s^\omega + s r) \cdot (\ell + H + s \log N))$
  operations by reinterpreting the product
  $J^r_\Delta(\OKe)^s \times (\OKe)^{s \times s} \to J^r_\Delta(\OKe)^s$
  as a matrix-matrix product
  $(\OKe)^{r \times s} \times (\OKe)^{s \times s} \to (\OKe)^{r \times s}$.

  Consequently, the cost of computing $P(r,N)$ from the~$B(n)$ is
  $\softO(s^{\omega} N (\ell + H + s))$
  and dominates that of constructing the~$B(n)$.
  Step~\ref{step:compute rec} takes
  $\softO(rs(\ell + s))$
  operations.
  Step~\ref{step:extract} takes
  $\softO(r^2(\ell + H + s + \sigma))$
  operations.
\end{proof}

When $| x | < \rho$, the series $y (x)$ converges geometrically, so $N =
\bigO (\sigma)$ terms suffice to attain the precision~$\bigO
(\pi^{\sigma})$. When additionally $x$ is the image in~$K$ of an algebraic
number of small bit size, one can take $H = \bigO (1)$ in
Proposition~\ref{prop:bs complexity}, and Algorithm~\ref{algo:bs ordinary}
is enough to evaluate $y (x)$ to the precision~$\bigO
(\pi^{\sigma})$ in time~$\softO (\sigma)$.

\subsection{The ``digit-burst'' method}
\label{ssec:dburst}

In general, though, computing $y
(x)$ for $x \in K$ to a precision $\bigO (\pi^{\sigma})$ requires
approximating $x$~itself by an element of $\Ke$ of height
about~$\sigma$, making the complexity bound roughly quadratic in~$\sigma$.
Chudnovsky and Chudnovsky~{\cite{ChudnovskyChudnovsky1988}} get around this
issue by using the analytic continuation formula
\[ \Phi_0 (x_0 + x_1 + \cdots + x_m) = \Phi_{x_0 + \cdots + x_{m - 1}} (x_m)
   \cdots \Phi_{x_0} (x_1) \Phi_0 (x_0) \]
along a path formed by approximations $x_0 + \cdots + x_m$ of~$x$ with an
exponentially increasing number of correct digits, balancing the speed
of convergence of the series with the height of the terms like in
Theorem~\ref{thm:log}.
Interestingly, the idea still applies in the $p$-adic case,
even though the ``analytic continuation'' process does not allow one 
to escape from the disk of convergence of $\Phi_0 (t)$.

We want to evaluate~$\Phi_0$ at a point~$x$ with $|x|_p < \rho$.
For simplicity, we limit ourselves to $x \in \OKe$ but the general
case can be handled in a similar fashion.
We first need some a priori bounds on the speed of convergence of series
solutions of~\eqref{eq:ODE}. 
Given $\rho > 0$, we denote by $\calA_\rho$ the subring of 
$K\llbracket t\rrbracket$ consisting of series $f(t) = \sum_{n \geq 0} a_n t^n$ for 
which the sequence $(|a_n|_p \rho^n)_{n \geq 0}$ is bounded from above.
The ring 
$\calA_\rho$ is equipped with the \emph{Gauss norm} $\Vert 
\cdot \Vert_\rho$ defined by
$\Vert {\textstyle \sum_{n \geq 0} a_n t^n} \Vert_\rho = 
\sup_{n \geq 0} |a_n|_p \rho^n$.
It satisfies the ultrametric triangle inequality 
$\Vert f + g \Vert_\rho \leq \max(\Vert f \Vert_\rho, \Vert g \Vert_\rho)$
and it is multiplicative (\emph{i.e.} $\Vert fg 
\Vert_\rho = \Vert f \Vert_\rho \Vert g \Vert_\rho$).
Geometrically, series belonging to $\calA_\rho$ converge on the \emph{open}
disk of center~$0$ and radius~$\rho$ and the Gauss norm corresponds to 
the sup norm on this disk taken over an algebraic closure.

\begin{proposition}
\label{prop:apriori}
For $f_0, f_1, \ldots, f_{r-1} \in \calA_\rho$ and for any
series $y \in K\llbracket t\rrbracket$ with
$y^{(r)} + f_{r-1} y^{(r-1)} + \cdots + f_1 y' + f_0 y = 0$,
one has:
\begin{align*}
y & \in \calA_{\tilde\rho} && \text{with} &
\tilde\rho &= \Rexp \cdot \min \Big( \rho, \,
\min_{0 \leq i < r} \Vert f_i \Vert_\rho^{\frac 1{i-r}}\Big), \\
\Vert y \Vert_{\tilde\rho} & \leq \tilde M && \text{with} &
\tilde M &= \max_{0 \leq i < r} \big|y^{(i)}(0)\big|_p,
\end{align*}
where we recall that $\Rexp = p^{-1/(p-1)}$.
\end{proposition}

\begin{proof}[Sketch of the proof]
Write $y = \sum_{n \geq 0} y_n t^n$. The coefficients $y_n$ satisfy a 
recurrence of the form~\eqref{eq:reccoeff}, except that the length
of the recurrence in now unbounded because the $f_i$ are series instead
of polynomials.
Using this recurrence, one checks by
induction on~$n$ that $|n!y_n|_p \leq \tilde M \cdot (\tilde\rho/
\Rexp)^n$. Using $|n!|_p \geq \Rexp^{-n}$, we obtain
$|y_n|_p \leq \smash{\tilde M \tilde\rho^n}$ and the proposition follows.
\end{proof}

Define $\slice (x, \sigma_0, \sigma_1)$ as the result of replacing by
zero the coefficients of $p^k$ in the $p$-adic expansion of the coordinates
of~$x$ except for $\sigma_0 \leqslant k < \sigma_1$, that is, using the
notation of \S \ref{ssec:exact subring}:
\begin{equation} \label{eq:slice}
  \textstyle
  \slice \left( \sum_{i, j} u_{i, j} X^i Y^j, \sigma_0, \sigma_1
  \right) = \sum_{i, j} \begin{aligned}[t]
  (&(u_{i, j} \mod p^{\sigma_1}) \\
  &- (u_{i, j} \mod p^{\sigma_0})) X^i Y^j.
  \end{aligned}
\end{equation}
Algorithm~\ref{algo:digit burst} implements the computation
of~$\Phi_0(x)$.

\begin{algorithm}[t]
\caption{%
$\mathrm{DigitBurstSolve} ( \boldsymbol a \in (\OKe)^r[t], \rho, M, x \in \OKe, \sigma )$
}
\label{algo:digit burst}
\begin{enumerate}[(1)]
  \item Let
  $\tilde\rho = \Rexp \cdot \min \Big( \rho, \,
   \min_{0 \leq i < r} \Big(\frac{\Vert a_i \Vert_\rho}
   {\Vert a_r \Vert_\rho}\Big)^{\smash{1/(i-r)}}\Big)$.
  \item Let
  $c = \lceil\max(\val(x), -1 + \log_p \tilde \rho) \rceil$,
  $\sigma' = \sigma - \min(0, \lfloor \log_p M \rfloor)$.
  \item Set $Y = \Id \in (\OKe)^{r \times r}$, $\boldsymbol a_{-1} = \boldsymbol a$, and $x_{-1} = 0$.
  \item For $m = 0, 1, \ldots$ while $c 2^m \leqslant \sigma$:
  \begin{enumerate}[(a)]
    \item Set $\boldsymbol a_m(X) = \boldsymbol a_{m-1} (x_{m-1} + X)$.
    \item Set $x_m = \slice (x, \, c \lfloor 2^{m - 1} \rfloor, \, c 2^m)$.
    \item If $m = 0$ then set $N_0 = (\sigma' + \log_p M)\,/\,\log_p(\rho/|x|_p)$,\\
          otherwise set $N_m = \sigma'\,/\,\log_p(\tilde \rho/|x_m|_p)$
    \item \label{step:burst:matrix}
    Set $Y = \tilde \Phi Y$ where $ \tilde \Phi =
    \mathrm{PartialSum}(\boldsymbol a_m, x_m, 1, \lceil N_m \rceil, \sigma')$.
  \end{enumerate}
  \item Return~$Y$.
\end{enumerate}
\end{algorithm}

\begin{lemma} \label{lem:taylor shift}
  If $a \in \OKe[t]$ is a polynomial of degree at most~$r$ with
  $h_K(a) \leq \ell$
  and $\xi \in \OKe$ has height $h_K(\xi) \leq H$,
  then $\tilde a = a(\xi + t)$ has height
  $h_K(\tilde a) \leq \ell + (r+1) H + \log r$.
\end{lemma}

\begin{proof}
  With $a = \sum_i c_i t^i$, one has
  $\tilde a = \sum_j \sum_i \smash{\binom{i}{j}} c_i \xi^{i-j} t^j$.
  Since $\log \smash{\binom{i}{j}} \leq i \leq r$,
  the claim follows by Proposition~\ref{prop:height}.
\end{proof}

\begin{proposition} \label{prop:bb complexity}
  Given $\boldsymbol a = (a_i) \in (\OKe)^r[t]$,
  $x \in \OKe$ with $h_K(x) \leq \sigma$ and $\rho, M \in \RR_{>0}$ such that:
  \begin{enumerate}[(1)]
  \item the leading coefficient $a_r$ of $\boldsymbol a$ is invertible in $\calA_\rho$
  (equivalently, all roots of $a_r$ in an algebraic closure have norm at least $\rho$),
  \item the entries of~$\Phi_0(t)$ lie in $\calA_\rho$ and
  have Gauss norm $\leq M$,
  \end{enumerate}
  Algorithm \ref{algo:digit burst} computes $\Phi_0(x) \mod \pi^\sigma$ in
  $\softO(s^\omega \sigma (\ell + s))$ operations.
\end{proposition}

\begin{proof}
  Consider iteration~$m$ of the loop.
  We have by construction
  $h_K(x_m) \leq c2^m$
  and
  $h_K(x_0 + \dots + x_{m-1}) \leq c2^{m-1}$.
  By Lemma~\ref{lem:taylor shift}, this implies
  $h_K(\boldsymbol a_m) \leq c (r+1) 2^{m-1} + \ell + \log r$.
  Using fast Taylor shift algorithms~\cite{GathenGerhard1997},
  one can compute the vector $\boldsymbol a_m$ from $\boldsymbol a_{m-1}$ in
  $\softO(cr^3 h_K(x_m) + cr^2 h_K(\boldsymbol a_{m-1}))
  = \softO(cr^3 2^m)
  = \softO(s^3 \sigma)$
  operations.
  Since, by~\eqref{eq:val} and~\eqref{eq:slice}, $|x_m|_p \leq p^{-c2^m}$,
  one gets $N_m = \bigO(c^{-1} 2^{-m} \sigma)$.
  By Proposition~\ref{prop:bs complexity}, the cost of the call to
  $\mathrm{PartialSum}$ is
  \[ \softO(s^\omega N_m (h_K(\boldsymbol a_m) + h_K(x_m) + s)
  = \softO(s^\omega \sigma (\ell + s)). \]
  As the number of iterations is $\bigO(\log \sigma)$, the total cost of the algorithm is $\softO(s^\omega \sigma (\ell + s))$.

  Regarding correctness, it follows from Proposition~\ref{prop:apriori}
  applied with $f_i(t) = \smash{\frac{a_i}{a_r}}(t{-}x_0{-}\cdots{-}x_{m-1})$
  and our choices of $\rho$, $M$
  and $\tilde \rho$ that the matrix $\tilde \Phi$
  computed at step~\ref{step:burst:matrix}
  is equal to $\Phi_{x_0 + \cdots + x_{m-1}}(x_m)$ modulo 
  $\pi^{\sigma'}$.
  Besides, the norm of its coefficients is
  bounded by $M$ when $m = 0$ and by $1$ otherwise. The product
  of all these matrices is then congruent to $\Phi_{x_0 + \cdots
  + x_m}(0) = \Phi_x(0)$ modulo $\pi^\sigma$.
\end{proof}

\subsection{Regular singularities}
\label{ssec:regsing}

For many interesting examples, the assumption $a_r(0) \neq 0$ is not
satisfied.
In this case, there may not exist a full basis of formal power
series solutions.
Series solutions that do exist still satisfy the
recurrence~\eqref{eq:rec} (whose order drops since~$b_0$ vanishes
identically),
but a solution~$y(t) \in K\llbracket t \rrbracket$ is not necessarily
characterized by its coefficients $y_0, \dots, y_{r-1}$, and may not
converge anywhere.

We focus here on the important special case where $0$~is a
\emph{regular singular point}, which means, by definition, that the
leading coefficient $Q_0 = b_{j_0}$ of~\eqref{eq:rec}, called the
indicial polynomial, has degree~$r$.
It is a classical fact~\cite[\emph{e.g.},][§16]{Poole1936}
that one can then construct $r$~linearly independent formal logarithmic
series solutions
\begin{equation}
  \label{eq:formal}
  f_j(t) = t^{\delta_j} \sum_{\nu=0}^{\infty} \sum_{k=0}^{\kappa_j-1}
           f_{j,\nu,k} \, t^\nu \, \frac{\log^k t}{k!},
  \qquad f_{j,\nu,k} \in K(\delta_j),
\end{equation}
where the $\delta_j$ are roots of~$Q_0$ in an algebraic closure 
of~$K$.
Letting~$E$ be the set of $(\delta,k)$ such that~$\delta$ is a root
of~$Q_0$ of multiplicity~$\mu(\delta)>k$,
the~$f_j$ can be chosen in such a way that, for each~$j$, exactly one of
the coefficients~$f_{j,\nu,k}$ for which $(\delta_j + \nu, k) \in E$ is
nonzero,
and one can take
$\kappa_j \leq \sum_{\delta_i-\delta_j \in \ZZ} \mu(\delta_i)$.
Moreover, by~\cite[Prop.~3]{Mezzarobba2010}, the relation~\eqref{eq:rec}
holds for the series~\eqref{eq:formal} when $f_{j,\nu}$ is
interpreted as the vector $(f_{j,\nu,k})_{0\leq k < \kappa_j}$ and $n$
is set to $\delta_j + \nu + \Lambda$ where $\Lambda$ is the operator
mapping $(c_k)_{k}$ to $(c_{k+1})_{k}$.
In other words, with $Q_j = b_{j_0+j}$ and $s' = s - j_0$, one has,
for all $j \in \{1, \dots, r\}$ and $\nu \in \ZZ$, \nopagebreak[4]
\begin{equation}
  \label{eq:sing rec}
  Q_0(\delta_j + \nu + \Lambda) (f_{j, \nu, k})_k
  + \dots
  + Q_{s'}(\delta_j + \nu + \Lambda) (f_{j, \nu - s', k})_k
  = 0.
\end{equation}

Making the bridge between these formal solutions and actual
analytic solutions is the subject of the Dwork-Robba theory of $p$-adic
exponents~\cite[\emph{e.g.},][\S 13]{Kedlaya2010}.
While the general case seems difficult to attack, when the exponents
$\delta_j$ all lie in $\Zp$, the formal expression~\eqref{eq:formal}
does define an analytic function on each ball of the form
$\{x_0(1+\xi) : |\xi|_p < 1\}$ with $|x_0|_p < \rho$ for a
suitable~$\rho>0$,
and binary splitting methods adapt, making it possible to evaluate the
fundamental matrix
$\Phi_0(t) = \bigl(\frac{1}{i!} f_{\smash j}^{\smash{(i)}}(t)\bigr)$
on any such ball in essentially linear time.
We must limit ourselves here to a succinct description of the algorithm, leaving for
future work a complete proof and complexity analysis (which however
proceed along the same lines as in~§\ref{ssec:ordinary},
see~\cite{Mezzarobba2011} for some details in the complex setting).

\begin{algorithm}[t]
\caption{%
$\mathrm{RegSingPartialSum}\bigl(
  (a_i), 
  u \in \OKe,
  v \in \mathbb{Z},
  N, 
  \sigma 
\bigr)$}
\label{algo:regsing}
\begin{enumerate}[(1)]
  \item \label{step:regsing:rec}
  Compute the polynomials $b_j$ defined by~\eqref{eq:rec},
  \eqref{eq:reccoeff}.\\
  Let $j_0 = \min\{j : b_j \neq 0\}$, $s' = s{-}j_0$ 
  and $Q_j = b_{j_0+j}$ for $0{\leq}j{\leq}s'$.
  \item Write
  $Q_0(n) = \zeta \prod_{q \in \mathcal Q} \prod_{\nu \in \mathcal N_q} q(n + \nu)^{m_{q, \nu}}$
  where $\zeta \in \Ke$, $\mathcal Q \subset \Ke[n]$ is a set of monic, irreducible,
  non-constant polynomials, $\mathcal N_q \subset \ZZ_{\geq 0}$,
  and any two distinct $q(n + \nu)$
  are coprime.
  \item Initialize $\tilde \Phi$ to an $r \times 0$ matrix over~$\Ke$.
  \item \label{step:regsing:loop}
  For $q \in \mathcal Q$:
  \begin{enumerate}
    \item \label{step:regsing:monic}
    Let $\gamma$ be the image of~$X$ in $\Ke[X]/q(X)$.
    Compute~$\tau \in \ZZ$ such that $\tau\gamma$
    is integral over $\OKe$.
    Let $\alpha = \tau\gamma$.
    \item Let $\kappa = \sum_{\nu \in \mathcal N_q} m_{q, \nu}$.\\
    Initialize $\Psi$ to an $(s'+1) \times 0$ matrix
    over~$\mathcal L_{\alpha,\kappa}$.
    \item For each pair $(\nu_0, \nu_1)$ of consecutive elements
    of~$\mathcal N_q \cup \{N\}$:
    \begin{enumerate}
    \item \label{step:regsing:shift}
    Right-shift all entries of~$\Psi$ by $m_{q, \nu_0}$, prepending zeros.
    \item
    For $k = 0, \dots, m_{q,\nu_0}-1$,
    append to~$\Psi$ a column of the form
    $(\text{[$s'{-}1$ zeros]}, L_{\alpha,\kappa}^k, 0)^{\mathrm T}$
    and attach to it the index~$\nu_0$.
    \item \label{step:regsing:binsplit}
    Compute $\Pi = \tilde P(\nu_0, \nu_1)$ by binary splitting.
    \item \label{step:regsing:act}
    Set $\Psi = \Pi \bullet \Psi$.
    \end{enumerate}
    \item Compute the set $\mathcal E$ of roots of~$q$ in~$\Zp$.
    Fail if $|\mathcal E| < \deg q$.
    \item \label{step:regsing:specialize}
    For each $c$ in the last row of~$\Psi$
    and each $\gamma^\ast \in \mathcal E$,
    append to~$\tilde \Phi$ the column vector of coefficients of
    $c(u/v, \tau\gamma^\ast, \gamma^\ast + \nu) \mod \pi^\sigma$,
    cf.~\eqref{eq:specialize}, where~$\nu$ is the index attached
    to the column of~$\Psi$.
  \end{enumerate}
  \item Return $\tilde \Phi$.
\end{enumerate}
\end{algorithm}

The procedure is summarized in Algorithm~\ref{algo:regsing}.
Its input is similar to that of Algorithm~\ref{algo:bs ordinary}, with
the understanding that the number of terms~$N$ to be computed now
applies separately to each set of solutions~$f_j$ whose
exponents~$\delta_j$ differ by integers.
(We assume for simplicity that
$N \geq \max (\{\delta_i - \delta_j\} \cap \ZZ)$.)
The differences with Algorithm~\ref{algo:bs ordinary} come from the need
to deal with exponents~$\delta_j$ lying in~$\Zp$ and with logarithmic
terms.

Exponents in~$\Zp$ cause no serious trouble.
The only subtlety is that, in order to keep the bit size of the
coefficients of~\eqref{eq:sing rec} small, 
we represent the~$\delta_j$ as elements of formal integral extensions
of~$\OKe$.
This is the role of step~\ref{step:regsing:monic}.
Computations performed in this representation are shared between
solutions~$f_j$ that are Galois conjugates of each other, which
roughly offsets the overhead of arithmetic in extension rings.
When no two exponents~$\delta_j$ differ by an element of~$\ZZ$,
all $\kappa_j$ are equal to~$1$,
and Algorithm~\ref{algo:regsing} reduces to something very similar to
Algorithm~\ref{algo:bs ordinary}, but slower by a factor~$\softO(r)$.

Logarithmic terms are dealt with by viewing the operator~$\Lambda$ as a
formal parameter in~\eqref{eq:sing rec} and inverting the leading
coefficient modulo~$\Lambda^{\kappa_j}$.
The main difficulty comes from zeros of~$Q_0$ that differ by integers,
leading to exceptional indices~$\nu_0$ where $Q_0(\delta_j + \nu_0 + \Lambda)$ is
not invertible in~$J^{\smash{\kappa_j}}_{\Delta}(\Ke)$.
Something special needs to be done to extend a
sequence~$(f_{j,\nu})_{\nu < \nu_0}$ past such a~$\nu_0$, whereas one
can see that \eqref{eq:sing rec} leaves the choice of $f_{j,\nu,k}$ for
$k < \mu(\delta_j + \nu)$ free, in accordance with the description of
the basis in terms of~$E$ above.

To make this more precise, let us focus on one
iteration of the loop starting at step~\ref{step:regsing:loop}.
We freely use the notation of the algorithm; in particular, $\alpha$ and
$\kappa$ are fixed.

In order to generalize the binary splitting algorithm of §\ref{ssec:ordinary}
to the new setting, we extend some of the components of~$\RecMat$ to have
$\mathsf{C} \in  J^{\kappa}_{\Lambda} ( \OKe [\alpha] ) ^{s'
\times s'}$,
$\mathsf{d} \in \OKe [\alpha]$,
and
$\mathsf{R} \in J^{\kappa}_{\Lambda} ( J^r_{\Delta} ( \OKe [\alpha] ) )^{s'}$,
the rest of the definition remaining formally the same.
For $\nu \in \ZZ_{\geq 0}$, we define
$\tilde b_0(\nu) \in \OKe[\alpha]$
and
$\tilde b_j(\nu) \in J^\kappa_\Lambda(\OKe[\alpha])$
by
$
  \tilde b_j(\nu)/\tilde b_0(\nu)
  =
  Q_j(\alpha + \nu +\Lambda)
  /(\Lambda^{-m_{q,\nu}} Q_0(\alpha + \nu +\Lambda))
$
with the convention $m_{q,\nu} = 0$ when $\nu \notin \mathcal N_q$.
This makes sense because, by definition, $\alpha + \nu$ has multiplicity
$m_{q,\nu}$ as a root of~$Q_0$.
Then we define $\tilde B(\nu)$ and $\tilde P(\nu_0, \nu_1)$ similarly to
$B(n)$ and $P(n_0, n_1)$ in §\ref{ssec:ordinary}, with
$s$ replaced by $s'$
and
each $b_j(n)$ replaced by $\tilde b_j(\nu)$.

We represent polynomials in $\log(t)$ appearing in coefficients and
partial sums of series using elements of
$\mathcal L_{\alpha, \kappa} = J^r_{\Delta}(\Ke[\alpha])^{\kappa}$.
The canonical basis of $\mathcal L_{\alpha,\kappa}$ over
$J^r_{\Delta}(\Ke[\alpha])$ is denoted $(L_{\alpha,\kappa}^k)_{k=0}^{\kappa-1}$.
Interpreting $\Lambda$ as the left-shift operator as above,
we obtain an action~$\bullet$ of
$J^\kappa_{\Lambda}(\Ke[\alpha])$ on~$\mathcal L_{\alpha,\kappa}$.
By identifying a tuple $T \in \RecMat$ with the matrix $\mathsf M_T$
and viewing the latter as a matrix over
$J_{\Lambda}^\kappa(J_{\Delta}^r(\OKe[\alpha]))$, 
it naturally extends to an action
of~$\RecMat$ on $(s'+1)$-row matrices with entries
in~$\mathcal L_{\alpha,\kappa}$.

With these conventions, one can check that
when~$Q_0(\gamma + \nu) \neq 0$,
applying $\tilde B(\nu)$ to a vector~$Y \in \mathcal L_{\alpha,\kappa}^{s'+1}$
that encodes~$s'$ consecutive terms of a solution and a corresponding
partial sum amounts to advancing to the next term using the
recurrence~\eqref{eq:sing rec}.
As in §\ref{ssec:ordinary}, the algorithm collects the $\tilde B(\nu)$
for $\nu$ between two roots of~$Q_0(\gamma + n)$ in a
product~$\Pi$ that is then applied to all
solutions whose computation is in progress.
When crossing a root~$\nu_0$ of $Q_0(\gamma + n)$, these solutions are
``shifted to the right'' (step~\ref{step:regsing:shift}) in a way that compensates
for the factor~$\Lambda^{m_{q, \nu_0}}$ missing in $\tilde B(\nu_0)$
compared to~$\eqref{eq:sing rec}$. 
``New'' solutions of
$t$-valuation $\gamma + \nu$ are added to the fundamental matrix.

Finally, at step~\ref{step:regsing:specialize}, the partial sums 
are converted to suitable specializations and are collected in a new 
matrix. More precisely, given $x = x_0 (1 + \xi)$ with $|x_0|_p < \rho$
and $|\xi|_p < 1$, we
define the specialization $c(x, \alpha^\ast, \delta) \in J^r_{\Delta}(\Ke)$ 
of $c = \sum c_k \smash{L_{\alpha,\kappa}^k} \in \mathcal L_{\alpha,\kappa}$
by
\begin{equation}
  \label{eq:specialize}
  c(x, \alpha^\ast, \delta)
  = x_0^\delta \: (1 + t)^{\delta} \cdot
    \sum_{k=0}^{\kappa-1} c_k(\alpha^\ast) \frac{\big(\log x_0 + \log(1+t)\big)^k}{k!}
\end{equation}
with $t = \xi + x_0^{-1} \Delta$.
Here
$c_k(\alpha^\ast)$ is the image of $c_k$ by the embedding
of~$J^r_{\Delta}(\Ke[\alpha])$ into~$J^r_\Delta(K)$ mapping $\alpha$ 
to~$\alpha^\ast$. 
The factors $(1+t)^{\smash\delta}$ and $\log(1+t)$ are given by converging series
and can be computed to
the precision~$\bigO(\pi^\sigma)$ in $\softO(\sigma)$ operations using
the algorithms of~§\ref{sec:elementary}. As for $x_0^\delta$ and $\log x_0$, 
they can be chosen almost arbitrarily, any choice corresponding to a 
valid branch of the solution.

As in §\ref{ssec:ordinary}, the main contribution to the cost is that of
step~\ref{step:regsing:binsplit}, and it is not too hard to see that
this step takes~$\softO(N)$ bit operations, all other parameters being
fixed.
The overhead of arithmetic in $J_\Lambda^\kappa(\OKe[\alpha])$, summed
over all~$\alpha$ and~$\kappa$, leads to an additional factor
$\softO(r)$ compared to Proposition~\ref{prop:bs complexity}
in the complexity of the full algorithm.
After using Algorithm~\ref{algo:regsing} to move away from a singularity, one can
continue with the digit-burst method (since the next steps fall under
the assumptions of~\S\ref{ssec:ordinary}), so that
Proposition~\ref{prop:bb complexity} adapts.

Formally, the algorithm also applies to partial sums of arbitrary
logarithmic series solutions, even at irregular singular points.
Only the existence of a full basis of the form~\eqref{eq:formal} and
its convergence properties depend on the regularity assumption.
In particular, if we know by external arguments that a certain
logarithmic series solution converges in a certain disk, we can
evaluate it by binary splitting and the digit-burst method without
trouble.

\section{Applications}
\label{sec:examples}

\subsection{Elementary and special functions}

\subsubsection*{Elementary functions.}

The methods of sections §\ref{ssec:ordinary} and~§\ref{ssec:dburst}
apply to the $p$-adic
logarithm and exponential as these functions both
satisfy simple differential equations. However, the 
specialized algorithms we presented in  \S \ref{sec:elementary}
perform much better in practice.
In contrast, the general power function and the Artin-Hasse 
exponential are not covered by these methods
because the differential equations annihilating 
them do not have small height degree.

\subsubsection*{Polylogarithms.}

Given a positive integer $s$, one can define the
$p$-adic polylogarithm function $\Li_s$ by
$\Li_s(t) = \sum_{i=1}^\infty \frac{t^i}{i^s}$.
This function
is solution to
$(1-t) \cdot D^{s+1}(y) = D^s(y) = 0$ (with $D = t{\cdot}\ddt$),
an equation with a regular singular point at the origin.
It can be evaluated in essentially linear time using the algorithms of
§\ref{ssec:dburst} and~§\ref{ssec:regsing}.

\subsubsection*{Gauss hypergeometric functions.}

Let $a$, $b$ and $c$ be three rational numbers with nonnegative
$p$-adic valuation and $c \not\in \ZZ^-$.
To these parameters, we associate
the Gauss hypergeometric function $\twoFone(a,b;c)$:
\begin{equation}
\label{eq:twoFone}
\twoFone(a,b;c;t) 
= \sum_{i=0}^\infty \frac{(a)_i (b)_i}{(c)_i \cdot i!} t^i
\end{equation}
where $(x)_i = x \cdot (x+1) \cdots (x+i-1)$.
The function $\twoFone(a,b;c)$ has radius of convergence $1$ and
satisfies the differential equation
$t(1-t)y'' - \big(c - (a+b+1)t\big) y' - ab y = 0$,
again with a regular singular point at the origin.
The algorithms of \S \ref{sec:ODE}
applied to $\twoFone(a,b;c;x)$ with $x \in K$, $|x|_p < 1$
run in essentially linear time for fixed $a,b,c$.

\subsection{Gauss-Manin connections}
\label{ssec:GMconnection}

The commutation of the Frobenius and the Gauss-Manin connection on the
cohomology of $p$-adic varieties gives rise to differential equations
with polynomial coefficients on the matrix of the Frobenius.
This results in a class of equations to which one can hope using the
methods of this paper to obtain interesting corollaries.

For an example of this phenomenon,
start with the Gauss hypergeometric function~\eqref{eq:twoFone}
of parameters $(a, b, c) = (\frac 1 2, \frac 1 2, 1)$
and consider the logarithmic derivative
\[ \textstyle f(t) = \twoFone'(\frac 1 2, \frac 1 2; 1; t) / 
\twoFone(\frac 1 2, \frac 1 2; 1; t).\]
This formula defines a series that converges on the open unit disk.
It turns out, however, that its sum~$f$ can be canonically extended to the
\emph{closed} unit disk as a consequence of~\cite[Lemma~3.1]{Dwork1969}.
Evaluating~$f$ at points of norm~$1$ is in principle difficult as the
series does not converge on the boundary.
Recently, though, Asakura~\cite{Asakura2020} and Kedlaya~\cite{Kedlaya2019}
independently noticed that values of~$f$ on the unit circle
appear in the cohomology of certain algebraic fibrations.

One can try to combine this beautiful observation with the techniques
of~§\ref{sec:ODE} to accelerate the computation of $f(x)$ when
$|x|_p = 1$ and $x \not\equiv 1 \pmod p$.
We conclude this paper with a short preview of results in this direction
that we plan to develop in a future extended version.

Let us first briefly review the main results of~\cite{Asakura2020}.
Let $X \subset \PP^1_x \times \PP^1_y \times \PP^1_t$ be the variety 
defined by the equation $(x^2 - 1) \cdot (y^2 - 1) = t$; we view it as 
a fibration over~$\PP^1_t$.
To this geometric situation, one can attach a cohomology space $H$ 
(the log-crystalline cohomology of~$X$), which is a 
module over $\Zpt$. For each choice of $c \in 1 + p \Zp$,
$H$ is equipped with a Frobenius map $\phi_c : H \to H$,
which is semi-linear in the sense that it is continuous, additive and
it satisfies $\phi_c(th) = ct^p \phi(h)$ for all $h \in H$.
Asakura shows that $H$ is a free module of 
rank~$2$ and exhibits a canonical basis of it.
Besides, he proves 
that, when $c = x^{1-p}$, the vector
$\bigl( x(x{-}1) f(x), 1 \bigr)^{\mathrm T}$
is the unique eigenvector of $\phi_c$ corresponding to an eigenvalue 
of norm~$1$. Thus, if we are able to compute $\phi_c$ (for $c = 
x^{1-p}$), we will be able to deduce the value~$f(x)$ we are
interested in.

For this, we use the so-called Gauss-Manin connection on $H$.
The Gauss-Manin
connection is a mapping $\nabla : H \to H \frac{\mathd t}t$ which encodes
the variation of the cohomology with the parameter $t$.
Writing that $\nabla$ commutes with $\phi_c$, we obtain the
following differential equation,
in which $M_c$ is the matrix of $\phi_c$ in Asakura's basis:
\begin{equation}
\label{eq:gaussmanin}
t \:\frac{\mathd M_c(t)}{\mathd t} + 
\left(\begin{matrix}
0 & \frac t 4 \\ \frac 1{t-1} & 0
\end{matrix}\right) M_c(t) - 
p M_c(t) 
\left(\begin{matrix}
0 & \frac {ct^p} 4 \\ \frac 1{ct^p-1} & 0
\end{matrix}\right)
= 0.
\end{equation}
Moreover it turns out that $M_c(t)$
overconverges outside the open unit disk and actually defines
an analytic function on the whole space punctured by the closed
disk of center $1$ and radius $\Rexp$.
Paying particular attention to the initial conditions, we can then 
use the methods of \S\ref{sec:ODE} to evaluate $M_c$ at any point
in the domain of convergence. Since this includes all
points $x$ with $|x|_p = 1$ and $x \not\equiv 1 \pmod p$, we have
reached our objective provided that $c = x^{1-p}$ is an integer
of small height.

Roughly speaking, what precedes corresponds to the first step in the 
digit-burst method. In order to handle the next steps, we come 
back to the hypergeometric differential equation. Indeed, fix $x_0 \in 
\Zp$ and let $G$ be the solution to the Cauchy problem
\begin{equation}
\label{eq:cauchy}
\begin{array}{l}
t(1-t)y'' + (2t-1) y' - \frac 1 4 y = 0, \smallskip \\
y(x_0) = 1, \; y'(x_0) = f(x_0).
\end{array}
\end{equation}
Then $G$ converges on the open disk of center~$x_0$ and radius $\Rexp$ 
and it follows by analytic continuation that $f(x) = 
G'(x)/G(x)$ on this domain. Thus, once we know the value 
of~$f(x_0)$, we can use~\eqref{eq:cauchy} to compute
$G(x)$ and $G'(x)$ by Algorithm~\ref{algo:digit burst}, and 
eventually recover the value of $f(x)$.

Putting both ingredients together, we end up with an algorithm
that evaluates $f(x)$ for $|x|_p = 1$ and $x \not\equiv 1 \pmod p$
with quasi-linear complexity in the output precision.
Note that the complexity with respect 
to~$p$ is not as good because the coefficients appearing in the 
differential equation \eqref{eq:gaussmanin} have degree of the order
of~$p$. The estimates of Proposition~\ref{prop:bb 
complexity} imply that the complexity in $p$ of our algorithm is 
in~$\softO(p^{\omega+1})$, which makes it practical for small values 
of~$p$ only. It would be interesting to try to lower this complexity by 
capitalizing on the sparsity of the polynomials appearing in 
\eqref{eq:gaussmanin}.

We have implemented part of the above algorithm
in SageMath, based on ore\_algebra%
\footnote{\url{https://github.com/mkauers/ore_algebra}, branch \texttt{padic}}.
Although it is still in development,
the application of it seems to be promising as the timing data
on Figure~\ref{table:dwork2F1} demonstrates.
Note that a naive evaluation of this function with precision $\sigma$
requires to evaluate series of $p^{\sigma}$ terms, e.g.,
$5^{20}=95367431640625$.
\begin{figure}
\begin{tabular}{rrr}
precision $\sigma$ & time (seconds)& value \\ \hline
12& 2.08& 1141554555 \\
16& 3.54& 468670851430 \\
20& 5.96& 372020184523305
\end{tabular}
\vspace*{-2ex}
\caption{\normalfont Log. derivative of $\twoFone(1/2,1/2;1;3^p)$ with $p{=}5$}
\vspace*{-0.5ex}
\label{table:dwork2F1}
\end{figure}

\paragraph*{Acknowledgements}
We thank the anonymous referees, whose comments
on a previous version of this paper
led to major presentation improvements.

\printbibliography

\end{document}